\documentclass[11pt,dvipsnames]{scrartcl}
\usepackage{titling}
\usepackage{amsfonts,amssymb,amsmath,amsthm,xspace}
\usepackage[sort,comma,numbers]{natbib}
\usepackage{enumitem}
\usepackage{tikz}
\usetikzlibrary{arrows}
\usepackage{mathtools}
\DeclarePairedDelimiter{\ceil}{\lceil}{\rceil}
\usepackage[font=small]{caption}
\usepackage{graphicx}
\usepackage{tikz}
\usetikzlibrary{calc}
\usetikzlibrary{positioning}
\pgfdeclarelayer{bg}    
\pgfsetlayers{bg,main}  

\usepackage[utf8]{inputenc}
\usepackage{hyperref}
\usepackage{cleveref}

\newtheorem{theorem}{Theorem}\newtheorem{corollary}[theorem]{Corollary}
\newtheorem{lemma}[theorem]{Lemma}

\newcommand{\dsp}[2]{$#1$-Disjoint Shortest Paths with Congestion-$#2$\xspace}
\newcommand{\sd}[1]{$#1$-Disjoint Shortest Paths\xspace}
\newcommand{\kdp}[1]{$#1$-Disjoint Paths Problem\xspace}

\newcommand{\disp}[2]{$#1$-Disjoint Paths with Congestion-$#2$\xspace}

\newenvironment{apthm}[1]{\par\addvspace{3mm}\noindent\textbf {Theorem~\ref{#1}.\;}}{\par\addvspace{3mm}}

\newenvironment{apcor}[1]{\par\addvspace{3mm}\noindent\textbf {Corollary~\ref{#1}.\;}}{\par\addvspace{3mm}}

\Crefname{lemma}{Lemma}{Lemmas}
\Crefname{figure}{Figure}{Figures}
\Crefname{corollary}{Corollary}{Corollaries}
\Crefname{definition}{Definition}{Definitions}

\title{Disjoint Shortest Paths with Congestion on DAGs}

\author{Saeed Akhoondian Amiri \and Julian Wargalla \thanks{Department of Computer Science, University of Cologne, Germany, \{amiri,wargalla\}@cs.uni-koeln.de}}

\date{}
\begin{document}	

	\maketitle	
	
	\begin{abstract}
	\noindent In the \sd{k} problem, a set of terminal pairs of vertices $\{(s_i,t_i)\mid 1\le i\le k\}$ is given and we are asked to find paths $P_1,\ldots,P_k$ such that each path $P_i$ is a shortest path from $s_i$ to $t_i$ and every vertex of the graph routes at most one of them. 		We introduce a generalization of the problem, namely, \dsp{k}{c} where every vertex is allowed to route up to $c$ paths. 	
	
	We provide a simple algorithm to solve the problem in time $f(k) n^{O(k-c)}$ on DAGs. Using the techniques for DAGs, we show the problem is solvable in time $f(k) n^{O(k)}$ on general undirected graphs. Our algorithm for DAGs is based on the earlier algorithm for \disp{k}{c}~\cite{amiriipl}, but we significantly simplify their argument.			 
	
	Then we prove that it is not possible to improve the algorithm significantly by showing that for every constant $c$ the problem is W[1]-hard w.r.t.\ parameter $k-c$. We also consider the problem on acyclic planar graphs, but this time we restrict ourselves to the edge-disjoint shortest paths problem. We show that even on acyclic planar graphs there is no $f(k)n^{o(k)}$ algorithm for the problem unless ETH fails. 	
	\end{abstract}
\thispagestyle{empty}
\clearpage
	\setcounter{page}{1}
	\section{Introduction}
	The \kdp{k} is one of the fundamental connectivity problems in graph theory. Given a graph $G$ and a set of source and terminal pairs $\{(s_i,t_i):1\le i\le k\}$ the goal is to connect every source $s_i$ to $t_i$ by internally vertex disjoint paths. The problem plays a central role in proving the graph minor theorem algorithmically~\cite{ROBERTSON199565}.
	
	One can relax the problem by allowing higher congestion on the nodes and edges of the graph. Such a generalization is relevant in practice: for instance, in communication networks, we can tolerate a certain amount of congestion. Routing disjoint paths with congestion have been extensively studied in the literature e.g., see~\cite{congestiongeneral,icalp,KawarabayashiKK14,ChekuriE13,ChekuriKS09}.	
	
	Another practical variation of the problem is to find disjoint paths that are shortest with respect to certain measures: one can require that every path has to be a shortest path, that the total length is minimized, that the maximum path length is minimized, etc. (see~\cite{shortest2,shortestplanar,shortest2g,planarshortest} for some of such cases).
	
	Most of the aforementioned variations add a criterion that makes the problem more difficult than the classical disjoint paths problem. However, one variation might make the problem easier: when every path connecting a source to a terminal should also be a shortest such path. This is known as the \sd{k} problem (\textbf{$k$-DSP}). 
	
	Not much was known about $k$-DSP since it had been raised as an open problem over 20 years ago~\cite{EILAMTZOREFF1998113}. The only known results were for the restricted case of two paths for undirected graphs~\cite{operationresearch,EILAMTZOREFF1998113} and  digraphs~\cite{shortest2}. Only recently, two separate studies provided algorithms with running time of $n^{f(k)}$~\cite{lochet,undirectedshortest2} for $k$-DSP on undirected graphs. The problem is wide open for general digraphs, even for the special case of $3$-DSP.	
	
	In this paper, we consider a generalized version of the $k$-DSP, namely the \dsp{k}{c} problem. This is similar to the $k$-DSP except that every vertex can tolerate a congestion $c$, that is, it can route up to $c$ paths.  
	
	We study the computational complexity of the problem on directed acyclic graphs (DAGs). DAGs are valuable in simulating scheduling problems. In such simulations, finding disjoint paths, and in particular disjoint paths of short length is quite important. For a nice example of such simulations, we refer the reader to the introduction of~\cite{dags}. 	Besides their practical applications, DAGs also form basic building blocks and can be used to study the theoretical aspects of general digraphs. Several digraph width measures are designed to measure similarity of the input graph to DAGs~\cite{dtw,dagwidth,dagwidth1,dwidth}.

	It is possible to devise an $n^{O(k)}$ algorithm for our problem on DAGs (e.g. see~\Cref{lem:congesteddsp}). However, we investigate the possibility of performing better when the congestion is close to $k$. The intuition is that if the congestion is $k$ we only have to route the shortest paths. We seek an algorithm with running time $f(k)n^{g(d)}$, where $d = k - c$ and $f$ and $g$ are computable functions.
	
	Since the problem is new, nothing is known on undirected graphs. We show it is possible to convert the existing works on $k$-DSP to \dsp{k}{c} to obtain the desired solution. In this case, we parameterize the problem by $k$ instead of $d$. The existing $k$-DSP algorithms are only relevant from a theoretical perspective since their running times are doubly and triply exponential~\cite{lochet,undirectedshortest2}. Since we do not improve the $k$-DSP algorithms we will reach a similar running time.
	
	\subsection*{Our Contribution and Comparison to the Previous Work}
	
We first show that the problem is solvable in polynomial time on DAGs for constant values of $k$ and $d$ by providing a $f(k) n^{O(d)}$ algorithm, i.e.\ we prove the following theorem.

	\begin{theorem}\label{thm:main}	The \dsp{k}{c} problem in acyclic graphs can be solved in time $f(k)n^{O(d)}$.\end{theorem}
	
Our algorithm is based on the algorithm for \disp{k}{c} on DAGs~\cite{amiriipl} and matches its running time.	Our technical contribution in that part is the major simplification of the analysis of the correctness of the existing algorithm. Let us elaborate on this. The analyses of both papers are based on a kernelization idea: find a small core of terminal pairs and argue that it is enough to just route those. 

The argumentation is based on the fact that we can reroute paths so that eliminating one terminal pair allows routing with less congestion. To find such a pair in the previous work the argument was that we can make a specific rerouting so that there is no high congested vertex. Adding some other conditions such a rerouting was called \emph{conservative} routing in~\cite{amiriipl}. 

The definition of such a rerouting added complexities to the analysis so that it was not possible to break that proof into smaller pieces. In the present work, we gradually improve a single path, not the entire solution and our argument is not on the global behavior of the solution. Thus, even though similar to previous work we employ the topological order of vertices, but we do not need to have such an order for the entire graph: it is enough to have a good order on certain rerouting points for that specific path. 

Our simplified method may open doors for dealing with similar problems on undirected graphs; since there we can define certain order for some set of vertices, not necessarily on the entire graph. Given the recent polynomial-time algorithm for $k$-disjoint shortest paths on undirected graphs~\cite{lochet,undirectedshortest2}, this subtle difference in analysis shows its importance.

\medskip

	We use one of the ingredients from the proof of the above theorem to show that every algorithm for $k$-DSP (in general graphs), transfers to the \dsp{k}{c} in linear time. However, in this transformation running time is dependent on $k$ and $n$.


\begin{corollary}\label{cor:transfer}
	Let $\mathcal{A}$ be an algorithm that solves an instance of $k$-DSP on (un)directed graphs in time $f(n,k)$, where $n$ is the size of the input graph. Then there is an algorithm $\mathcal{B}$ that solves an instance of \dsp{k}{c} on (un)directed graphs in time $f(n,k) + O(cn)$. In particular, for any fixed $k$, \dsp{k}{c} is solvable on undirected graphs in polynomial time.
\end{corollary}


The next natural question is whether the running time of our algorithm for DAGs is tight? We answer this almost affirmatively, by proving the following theorem.

\begin{theorem}
	\label{thm:hardness}
	For any fixed $c$, the \dsp{k}{c} problem cannot be solved in time $f(k) n^{o(k / \log k)}$ time, unless ETH fails.
\end{theorem}

To prove the above theorem we show that the already existing hardness proof for \disp{k}{c} extends to \dsp{k}{c}. To achieve this, we show either every path in the existing hardness construction is a shortest path or it can be converted to a shortest path.

\medskip
In the next step, we consider the problem on more restricted graph families, in particular acyclic planar graphs. However, in this case, we consider the $k$-Edge Disjoint Shortest Paths ($k$-EDSP) variant of the problem and only for congestion $1$. We prove the following theorem.

\begin{theorem}\label{thm:sdpplanar}
Unless ETH fails, the $k$-EDSP problem does not admit any $f(k)n^{o(k)}$ algorithm even on planar DAGs.
\end{theorem}

This matches the recently improved lower bound for disjoint paths on planar DAGs by Chitnis~\cite{chitnis} and suggests that maybe the two problems are in a similar hardness category. On the other hand, the complexity of vertex disjoint shortest paths on planar DAGs remains open; even it is not clear if the problem is NP-hard, while it is known that the planar vertex disjoint paths on planar DAGs (even on upward planar graphs) is NP-hard~\cite{csr14}.

		\subsection*{Prerequisites and General Definitions}
	In this work, we merely require a basic familiarity with graph theory, and parameterized complexity, for both concepts we refer the reader to the book~\cite{parametrized}. 
	
	An instance $I$ of the problem is composed of the input graph $G=(V,E)$ and an integer $c$ as congestion and a set of $k$ source terminal pairs $\{(s_1,t_1),\ldots,(s_k,t_k)\}\subseteq V\times V$. When it is clear from context, we may only refer to source terminal pairs of $I$. In addition, we may further abbreviate it and write \emph{terminal pairs}. Every edge may have a positive integer as its weight (or $1$ by default); $n$ denotes the number of vertices of $G$.

	\section{Algorithm}
	First, let us recall the approach of~\cite{amiriipl} for disjoint paths with congestion, then we introduce our intuitive analysis for disjoint shortest paths with congestion.
	The general idea of~\cite{amiriipl} was inspired by kernelization techniques: reduce the instance to smaller instances with congestion $1$. In particular, sub-instances with at most $O(d)$ pairs (recall that $d=k-c$). Since we know how to solve the $d$-Vertex Disjoint Paths problem in time $n^{O(d)}$ we can solve the \disp{k}{c} in almost the same running time.
	
	The main challenge in the previous work was the reduction to a smaller instance. We provide such a reduction in the following lemma, which is analogous to Lemma 8 of~\cite{amiriipl}. Our argument is simpler and intuitive, in addition, directly applies to both problems.	In the following everything goes the same way even if we have edge weights since every existing path in the following is a shortest path.	
	
	\begin{lemma}\label{lem:reduction}Let $k>3d$, then there is a solution $S$ to an instance $I$ of \dsp{k}{c} in a DAG $G$ if and only if the following two conditions hold: 
		
		\begin{itemize}[topsep=0pt]
			\itemsep-0.3em 
			\item there is a path in $G$ between each source terminal pair of $I$,
		\item there is a set $\mathcal{C}$ of $3d$ terminal pairs of $I$ such that \dsp{3d}{2d} has a solution for the terminal pairs of $\mathcal{C}$.
				\end{itemize}
	\end{lemma}
	\begin{proof}
		We assume the first condition is fulfilled, otherwise we are done.
		
		If a solution to $\mathcal{C}$ with congestion $2d$ is given then we can route the rest of $k-3d$ paths arbitrarily by their shortest paths and the congestion does not exceed $k-3d + 2d = c$. 
		
		We may assume $S$ has a vertex of congestion $c$ otherwise drop a pair from $I$ and the corresponding path from $S$ to work on a smaller instance (update $c=c-1,k=k-1$). As long as $k> 3d$ and $S$ has no vertex of congestion $c$ we repeat the above process.  
		The following claim is our main contribution to show the reverse implication of lemma.

	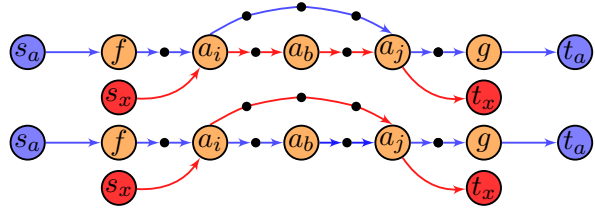
\begin{figure}
	\centering
	\begin{tikzpicture}[scale=1.2, every node/.style={fill=black, draw=black,text=black,circle,inner sep=0pt, minimum size=0.1cm},every path/.style={line width=0.25mm,draw=black}]

	\node[minimum size=0.45cm, fill=orange!60] (f) at (1,0) {$f$};
	\node[minimum size=0.45cm, fill=orange!60] (ai) at (2,0) {$a_i$};
	\node[minimum size=0.45cm, fill=orange!60] (b) at (3,0) {$a_b$};
	\node[minimum size=0.45cm, fill=orange!60] (aj) at (4,0) {$a_j$};
	\node[minimum size=0.45cm, fill=orange!60] (g) at (5,0) {$g$};
	
	\node (a1) at (2.4,0.4) {};
	\node (a3) at (3.6,0.4) {};
	\node (a2) at (3,0.5) {};
	\node[minimum size=0.45cm, fill=red!80] (s1) at (1,-0.5) {$s_x$};
	\node (s2) at (1.5,0) {};
	\node (s3) at (4.5,0) {};
	\node[minimum size=0.45cm, fill=red!80] (s4) at (5,-0.5) {$t_x$};
	\node[minimum size=0.45cm, fill=blue!50] (s5) at (0,0) {$s_a$};
	\node[minimum size=0.45cm, fill=blue!50] (s6) at (6,0) {$t_a$};
	\node (s7) at (2.5,0) {};
	\node (s8) at (4.5,0) {};
	\node (s9) at (3.5,0) {};

	\node[minimum size=0.45cm, fill=orange!60] (f1) at (1,-1) {$f$};
	\node[minimum size=0.45cm, fill=orange!60] (ai1) at (2,-1) {$a_i$};
	\node[minimum size=0.45cm, fill=orange!60] (b1) at (3,-1) {$a_b$};
	\node[minimum size=0.45cm, fill=orange!60] (aj1) at (4,-1) {$a_j$};
	\node[minimum size=0.45cm, fill=orange!60] (g1) at (5,-1) {$g$};
	
	\node (a11) at (2.4,-0.6) {};
	\node (a31) at (3.6,-0.6) {};
	\node (a21) at (3,-0.5) {};
	\node[minimum size=0.45cm, fill=red!80] (s11) at (1,-1.5) {$s_x$};
	\node (s21) at (1.5,-1) {};
	\node (s31) at (4.5,-1) {};
	\node[minimum size=0.45cm, fill=red!80] (s41) at (5,-1.5) {$t_x$};
	\node[minimum size=0.45cm, fill=blue!50] (s51) at (0,-1) {$s_a$};
	\node[minimum size=0.45cm, fill=blue!50] (s61) at (6,-1) {$t_a$};
	\node (s71) at (2.5,-1) {};
	\node (s81) at (4.5,-1) {};
	\node (s91) at (3.5,-1) {};

	\tikzset{edge/.style = {->,> = latex'}}
	
	\begin{pgfonlayer}{bg}
	\draw[edge,color=blue!70] (s5) to (f);	
	\draw[edge,color=blue!70] (f) to (s2);
	\draw[edge,color=blue!70] (s2) to (ai);
	\draw[edge,color=red!90] (ai) to (s7);
	\draw[edge,color=red!90] (s7) to (b);
	\draw[edge,color=red!90] (b) to (s9);
	\draw[edge,color=red!90] (s9) to (aj);
	\draw[edge,color=blue!70] (aj) to (s8);
	\draw[edge,color=blue!70] (s8) to (g);
	\draw[edge,color=blue!70] (g) to (s6);
	\draw[edge,color=blue!70] plot [smooth] coordinates {(ai.north) (a1) (a2) (a3) (aj.north)};
	%
	\draw[edge,color=red!90] (s1) to [bend right] (ai);
	\draw[edge,color=red!90] (aj) to [bend right] (s4);

	\draw[edge,color=blue!70] (s51) to (f1);	
	\draw[edge,color=blue!70] (f1) to (s21);
	\draw[edge,color=blue!70] (s21) to (ai1);
	\draw[edge,color=blue!70] (ai1) to (s71);
	\draw[edge,color=blue!70] (s71) to (b1);
	\draw[edge,color=blue!90] (b1) to (s91);
	\draw[edge,color=blue!90] (s91) to (aj1);
	\draw[edge,color=blue!70] (aj1) to (s81);
	\draw[edge,color=blue!70] (s81) to (g1);
	\draw[edge,color=blue!70] (g1) to (s61);
	\draw[edge,color=red!90] plot [smooth] coordinates {(ai1.north) (a11) (a21) (a31) (aj1.north)};
	
	\draw[edge,color=red!90] (s11) to [bend right] (ai1);
	\draw[edge,color=red!90] (aj1) to [bend right] (s41);
	\end{pgfonlayer}
	
	\end{tikzpicture}
	\caption{Swapping operation: $P_a$ is in blue and $P_x$ is in red. High congested vertices are colored orange (only $2$ paths are depicted). The top pair, shows the initial state of the two paths, the bottom pair shows them after swapping the middle subpaths. Since $a_i,a_j$ were the closest possible high congestion vertices of $P_a$ to/from $a_b$, the updated $P_a$ ($P^1_a$) has more high congestion vertices than the old one.}
	\label{fig:reroute}
\end{figure}
		\medskip
		\noindent\textbf{Claim.}
			If $S$ has a vertex of congestion $c$ then $I$ has a solution $S'$ such that:
			\begin{itemize}
				\itemsep-0.3em 
				\item congestion of every vertex is the same in $S$ and $S'$ and,
				\item there is a path $P\in S'$ that contains all vertices of congestion $c$.
			\end{itemize}   
		\medskip
		
		Let us first show how the lemma follows from the claim. 
		If the claim is correct then 
		 remove the endpoints of $P$ from $I$ to obtain $I'$. Then $S''=S'-P$ is a solution for $I'$. Afterwards update $k=k-1, c=c-1, S=S'', I=I'$ and as long as $k>3d$ repeat the same process. Once we have an instance with $3d$ pairs, the lemma follows.
		
		\medskip
		
		\noindent\textbf{Proof of Claim. }Order vertices of congestion $c$ in $S$ w.r.t.\ their topological order in $G$ to obtain a strictly increasing ordered set of them $\mathcal{A}=(a_1,\ldots,a_\ell)$. If there is a path $P_i\in S$ that visits all vertices of $\mathcal{A}$ we are done; otherwise, we need the following observation $\bigodot$ from the proof of Lemma 8 in~\cite{amiriipl}:
		
		\noindent\textbf{Observation $\bigodot$ }\label{starstar} For every (at most) $3$ vertices $a_i,a_j,a_h \in \mathcal{A}$ there is a path of $S$ that contains all of them: since $k>3d$, $a_i,a_j$ share at least $d+1$ common paths and since congestion of $a_h$ is $k-d$ it routes at least one of these $d+1$ paths.
		\medskip
		
		Note that if $\ell \le 3$ then there is a path that contains all vertices of $\mathcal{A}$ and we are done, hence in the rest we assume $\ell>3$. 
		There is a path $P_{a}\in S$ ($P_a$ connects the $a$'th terminal pair) that contains $a_1,a_\ell$ but not a vertex $a_b\in \mathcal{A}$. We reroute paths of $S$ so that the updated $P_a$, we call it $P^1_a$, contains $a_b$.
		
		Choose $a_i,a_j \in \mathcal{A}\cap P_a$ such that $a_i<a_b<a_j$ (w.r.t.\ their order in $\mathcal{A}$) and additionally they are closest such high congestion pair of vertices in $P_a$ surrendering $a_b$; i.e., for every other pair of vertices $a_f,a_h \in \mathcal{A}\cap P_a$ if $a_f<a_b<a_h$ then $a_f\le a_i$ and $a_j\le a_h$.
		
		By $\bigodot$ there is a path $P_x\in S$ such that $a_i,a_b,a_j \in V(P_x)$.
		 Replace a subpath of $P_x$ that starts and ends at $a_i,a_j$ with a subpath of $P_a$ that starts and ends at $a_i,a_j$ resp., to obtain paths $P'_x,P^1_a$. Let define $S^1=(S-\{P_a,P_x\})\cup \{P^1_a,P'_x\}$. 
		By choice of $a_i,a_j$, we have $\mathcal{A}\cap V(P_a) \subseteq V(P^1_a)$. Additionally, $P^1_a$ contains at least a new high congestion vertex $a_b\notin V(P_a)$. The congestion of all vertices and the path lengths in $S_1$ are the same as in $S$. Thus, $S^1$ is a solution to $I$. See~\Cref{fig:reroute} for an illustration of rerouting.
		
		We repeat the same process to construct paths $P^2_a,P^3_a,\ldots$; since each $P^i_a$ has more high congestion vertices than $P^{i-1}_a$, eventually we construct a path $P^t_a$ (for some $t\le \ell -2$) in a solution $S'=S^t$ that contains all vertices of congestion $c$ as claimed.\end{proof}	
	
	To find a feasible routing for a small set of terminal pairs in~\cite{amiriipl}, the algorithm of Fortune et al.~\cite{FORTUNE1980111} was employed. Kobayashi and Sako~\cite{shortest2} extended the construction of Fortune et al., to the shortest disjoint paths on DAGs. In those constructions, they use an auxiliary DAG structure that makes the algorithms none intuitive. Here we provide a much simpler algorithm that basically uses an elementary divide and conquer algorithm.   
	
	\newcommand{\head}[1]{\operatorname{Head}(#1)}
	
	\newcommand{\tail}[1]{\operatorname{Tail}(#1)}
	
	For a directed edge $e=(u,v)$ ($u\rightarrow v$), its tail is $u$ and $v$ is its head; we write $u=\tail{e},v=\head{e}$.
	
	\begin{lemma}\label{lem:shortestdisjointpathsdag}The $k$-DSP problem on DAGs can be solved in time $n^{O(k)}$.
	\end{lemma}\begin{proof}
		
		We solve a more general version of $k$-DSP problem: 
		\noindent For every $h$ tuple $H$ of source and terminal pairs, $0<h\le k$, we solve $h$-DSP between them. If there is a solution $S$ for $H$, store it in a dictionary $\mathcal{D}$; i.e.\ $\mathcal{D}[H]=S$. otherwise set $\mathcal{D}[H]=\emptyset$.
		
		To solve the above problem,  first topological sort vertices of the graph as $v_1,\ldots,v_n$. Afterward, recursively solve the problem in subsets of vertices $V_1=\{v_1,\ldots,v_{\ceil{n/2}}\}$ and $V_2=\{v_{\ceil{n/2}+1},\ldots,v_n\}$ for every $h$-tuple in these subsets ($1\le h\le k$). 	
		Then any $h$-tuple $H$ for the union of two sets falls in one of the following two categories ($h\le k$): 
		\begin{itemize}
			\itemsep-0.3em 
			\item $H$ can be partitioned into two tuples $H_1,H_2$ where $H_1\subseteq V_1\times V_1$ and $H_2\subseteq V_2\times V_2$.
			\item There is a source terminal pair $(s_i,t_i)\in H$ such that $s_i\in V_1, t_i \in V_2$. 
		\end{itemize}
		For the first case if $H_i$ ($i\in \{1,2\}$) is not empty but there is no solution for it, i.e., $\mathcal{D}[H_i]=\emptyset$, then set $\mathcal{D}[H] = \emptyset$ otherwise set $\mathcal{D}[H]=\mathcal{D}[H_1]\cup\mathcal{D}[H_2]$.
		
		It remains to solve the second case. 
		Suppose $t\ge 1$ terminal pairs of $H$ have one end in $V_1$ and the other end in $V_2$, let us call the set of them $H'$. 
		
		If there is a solution $S$ for $H$ then there are $t$ paths in $S$ going along an ordered set of $t$ edges $\mathcal{E}=(e_1,\ldots,e_t)$ where $\tail{e_i}\in V_1,\head{e_i}\in V_2$. 
		We do not know the endpoints of these edges but we can guess them.
		
		 Afterwards create two new sub-instances $H_1,H_2$: $H_1$ contains all source and terminal pairs that are entirely in $V_1$, in addition it has $t$ additional source terminal pairs where their source vertices are the sources in $H'$ and the terminal vertices are the tail of corresponding edges (according to their order) in $\mathcal{E}$.
		Similarly $H_2$ has all source and terminal pairs of $H$ that are entirely in $V_2$ and a new set of $t$ source terminal pairs where their sources are the heads of edges in $\mathcal{E}$ and their terminals are the terminals of $H'$. For an illustration of $H_1,H_2$ and the edge set~$\mathcal{E}$ see~\Cref{fig:merg}.
		
		In time $O(k)$ we can check if $\mathcal{D}[H_1],\mathcal{D}[H_2]$ are non-empty and their concatenation with edges of $\mathcal{E}$ yields a valid solution $S$ for $H$. If it is so, we set $\mathcal{D}[H]=S$. If no such set of edges $\mathcal{E}$ provides a desired solution for $H$, we set $\mathcal{D}[H]=\emptyset$.
		
		Since we considered all possibilities of $h$-tuples of the terminals, at the end of the algorithm, for a given tuple of source terminal pairs, we pick its value from $\mathcal{D}$.
		
		The running time of each recursive step is dominated by the merge operation, which can be done in $O(kn^{3k})$: there are $O(n^k)$ $h$-tuples, for each of them we guess $O(n^{2k})$ ordered sets of at most $k$-edges, and  in $O(k)$ we compare each path length (summation of $3$ values) with the corresponding value of the shortest path matrix. Since the input set of vertices in each call is partitioned, the total number of merge steps is $O(n)$, hence the running time is as claimed.
	\end{proof}

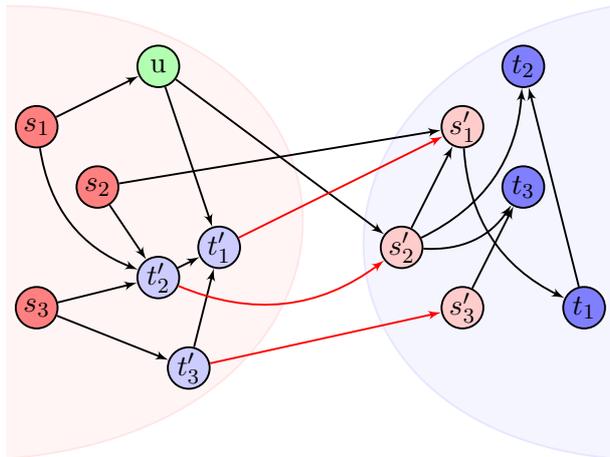
\begin{figure}
	\centering
	\begin{tikzpicture}[scale=0.8,every node/.style={draw=black,text=black,circle,fill=green!30,inner sep=0pt, minimum size=0.55cm},every path/.style={line width=0.25mm,draw=black}]

	\draw[color=red!10, fill=red!4] (-0.5,7) .. controls (6,7) and (6,0) ..(-0.5,-0.5);
	
	\draw[color=blue!10, fill=blue!4] (9.5,7) .. controls (4,6) and (4,0) .. (9.5,-0.5);

	\node[fill=red!50] (s1) at (0,5) {$s_1$};
	\node[fill=red!50] (s2) at (1,4) {$s_2$};
	\node[fill=red!50] (s3) at (0,2) {$s_3$};
	
	\node[fill=red!20] (s11) at (7,5) {$s'_1$};
	\node[fill=red!20] (s21) at (6,3) {$s'_2$};
	\node[fill=red!20] (s31) at (7,2) {$s'_3$};

	\node[fill=blue!20] (t11) at (3,3) {$t'_1$};
	\node[fill=blue!20] (t21) at (2,2.5) {$t'_2$};
	\node[fill=blue!20] (t31) at (2.5,1) {$t'_3$};

	\node[fill=blue!50] (t1) at (9,2) {$t_1$};
	\node[fill=blue!50] (t2) at (8,6) {$t_2$};
	\node[fill=blue!50] (t3) at (8,4) {$t_3$};
	
	\node (u) at (2,6) {u};

	\tikzset{edge/.style = {->,> = latex'}}
	\draw[edge] (s1) to (u);
	\draw[edge] (u) to (t11);
	\draw[edge] (s1) to[bend right] (t21);
	\draw[edge] (s2) to (t21);
	\draw[edge] (s3) to (t21);
	\draw[edge] (s3) to (t31);
	\draw[edge] (t31) to (t11);
	\draw[edge] (t21) to (t11);
	\draw[edge] (u) to (s21);		
	\draw[edge,color=red] (t11) to (s11);
	\draw[edge] (s11) to [bend right] (t1);
	\draw[edge,color=red] (t31) to (s31);
	\draw[edge,color=red] (t21) to [bend right] (s21);
	\draw[edge] (s31) to (t3);
	\draw[edge] (s21) to [bend right] (t2);
	\draw[edge] (s21) to [bend right] (t3);
	\draw[edge] (s2) to (s11);
	\draw[edge] (t1) to (t2);
	\draw[edge] (s21) to (s11);
	
	\end{tikzpicture}
	\caption{An illustration of the merge operation described in~\Cref{lem:shortestdisjointpathsdag}. The instance $H$ has $3$ source and terminal pairs. The set of $3$ (boundary) edges that are in the shortest paths connecting the sources to the terminals are colored red. The two intermediate instances $H_1,H_2$ are depicted in the left and right respectively.}
	\label{fig:merg}
\end{figure}
	
	The following is similar to Lemma 6 of~\cite{amiriipl}, except that we use~\Cref{lem:shortestdisjointpathsdag} instead of algorithm of Fortune et al. On the other hand since we prove the following lemma in its most general form later in~\Cref{cor:transfer}, for the moment we omit its proof. 
	
	\begin{lemma}\label{lem:congesteddsp}There is an algorithm that solves the \dsp{k}{c} in time $n^{O(k)}$ on DAGs.\end{lemma}
	
	Now our main theorem is a consequence of previous lemmas.
	
	\begin{apthm}{thm:main}
		The \dsp{k}{c} problem in acyclic graphs can be solved in time $f(k)n^{O(d)}$.
	\end{apthm}
	
	\begin{proof}
		If $k\le 3d$ then we directly apply~\Cref{lem:congesteddsp}, otherwise, by~\Cref{lem:reduction} we only need to guess $3d$ source terminal pairs (there are $k \choose 3d$ such choices) then route them in time $n^{O(d)}$ using~\Cref{lem:congesteddsp}. Afterwards connect the remaining pairs via their shortest paths. If the algorithm fails in any phase, there is no solution to the given instance\end{proof}

One of the simple but powerful tools that we used is the~\Cref{lem:congesteddsp}, the lemma actually applies on general digraphs which helps to prove the following corollary. In the following we show how it generalizes to general digraphs.

\begin{apcor}{cor:transfer}
	Let $\mathcal{A}$ be an algorithm that solves an instance of $k$-DSP on (un)directed graphs in time $f(n,k)$, where $n$ is the size of the input graph. Then there is an algorithm $\mathcal{B}$ that solves an instance of \dsp{k}{c} on (un)directed graphs in time $f(n,k) + O(cn)$. In particular, for any fixed $k$, \dsp{k}{c} is solvable on undirected graphs in polynomial time.
\end{apcor}
\begin{proof}
	 
We may assume that no terminal serves as an internal vertex for other paths in any solution for a given instance $I$ of $k$-DSP. To justify this assumption, note that for every terminal pair $s_i,t_i$  we can add new terminal pairs $s'_i,t'_i$ and edges $(s'_i,s_i), (t'_i,t_i)$ to ensure other paths would not go through $s'_i,t'_i$ (new terminal pairs). Thus, from now on we assume no terminal pair routes any other path.
	
Let us suppose that an instance $I$ of \dsp{k}{c} is given, with the underlying graph being denoted by $G$. We construct a graph $G'$ from $G$ by first copying every non-terminal vertex $v$ for $c$ times as $v^1,\ldots,v^c$; then for every edge $(u,v)\in E(G)$ add edges $(u^i,v^i)$ to $G'$ (with the same weight as $(u,v)$). Our instance $I'$ of $k$-DSP is the graph $G'$ together with the original source terminal pairs. We prove that the \dsp{k}{c} for $I$ has a solution if and only if $k$-DSP in $G'$ with terminal pairs $\{(s_i,t_i):1\le i\le k\}$ has a solution. 

If there is a solution for the instance of $k$-DSP in $G'$, then by merging the $c$ copies of every vertex into its original form, the corresponding (possibly merged) paths will result in congestion at most $c$ on every vertex. On the other hand, since the paths were the shortest path and we did not change the length of the edges in the process, every path in the resulting solution is also a shortest path. 

For the other direction, suppose a solution to \dsp{k}{c} in $G$ is given as a set of paths $\mathcal{P}=\{P_1,\ldots,P_k\}$. Let suppose a path $P_i=s_i,v_1,\ldots,v_\ell,t_i$. Then in $G'$ for terminal pair $s_i,t_i$ we construct a path $P'_i=s_i,v^i_1,\ldots,v^i_\ell,t_i$. First of all observe that $P'_i$ is a shortest path connecting $s_i,t_i$ otherwise $P_i$ was not a shortest path in $G$. We claim for $i\neq j \colon P'_i\cap P'_j = \emptyset$. This is actually by definition, since terminals are distinct, and for path $P'_i$ we used $i$'th copy of $v$ i.e.\ vertices of form $v^i_x$ which are distinct from the $j$'th copies of $v$ that has been used in $P'_j$.

Now suppose the algorithm $\mathcal{A}$ as stated above is given. Let $I$ be an instance of \dsp{k}{c}. In time $O(c|G|)$ we convert $I$ to an instance $I'$ as explained above. Then we solve instance $I'$ by algorithm $\mathcal{A}$ and at the end convert its output to the output of $I$ in time $O(c|G|)$. In particular since for undirected graphs there are known algorithms with running time $n^{f(k)}$ for $k$-DSP, it follows that the congested version is solvable in time $n^{f(k)}$.
\end{proof}

\section{Hardness}

%
%

For the $k$-DSP with congestion, to provide a lower bound with dependency on $d$, one can either use the Slivkins~\cite{Slivkins10} construction or the more modern construction of~\cite{amirimfcs}\footnote{We cited the conference version here since it is public and this part is the same as the journal version.} and then argue either every path is already a shortest path or it can be converted to a shortest path (even though Slivkin did not write his paper on routing with congestion, it is possible to modify his construction to reflect the congestion). In this work, we use the latter, since it provides a better lower bound (Slivkins' construction can only get us to $n^{o(\sqrt{k})}$ lower bound, while the more recent result, gives us $n^{o(k/\log k)}$ lower bound). Let us first briefly recall the hardness construction of~\cite{amirimfcs}, then we prove our first hardness result. 

\medskip
\noindent\textbf{Hardness for constant $c$}

In~\cite{amirimfcs}, the authors show that the \disp{k}{c} problem is relatively hard by relating it to partitioned subgraph isomorphism. In the latter problem one is given two graphs $H$ and $G$ together with an enumeration $V(H) = \{u_1, \ldots, u_k\}$ and a partitioning $V(G) = V_1 \dot{\cup} \ldots \dot{\cup} V_k$. The task is to find a graph homomorphism $\varphi: H \to G$, such that $\varphi(u_i) \in V_i$. This problem cannot be solved in time $f(k) n^{o(k / \log k)}$, where $k = |E(H)|$ and $n = |V(G)|$, even if $H$ is 3-regular and bipartite, unless ETH fails \cite{marx}.

This hardness result for partitioned subgraph isomorphism problem carries over to the \disp{k}{c} problem via an appropriate reduction \cite{amirimfcs}. That is, assuming ETH, \dsp{k}{c} cannot be solved in $f(k) n^{o(k / \log k)}$ time for constant $c$. We show that their reduction still holds, even if all paths are required to be shortest paths.

	\begin{apthm}{thm:hardness}
	For any fixed $c$, the \dsp{k}{c} problem cannot be solved in time $f(k) n^{o(k / \log k)}$ time, unless ETH fails.
\end{apthm}

\pgfdeclarelayer{demands}
\pgfdeclarelayer{graph}
\pgfsetlayers{demands,main,graph}

\begin{figure}

	\centering
	
	\begin{tikzpicture}[label distance=-1.5mm,every node/.style={draw=black,text=black,circle,fill=black,inner sep=0pt, minimum size=0.15cm},every path/.style={line width=0.25mm,draw=black}]
		
		\tikzset{edge/.style = {->,>= latex'}}
		
		\begin{pgfonlayer}{graph}
		
			\node[label={[label distance=-0.2mm]90:\tiny $s_\ell$}] (sl) at (4.25,2.25) {};
			\node[label={[label distance=0.2mm]270:\tiny $t_\ell$}] (tl) at (4.25,-5.75) {};
			
			
			\node[label={[label distance=-0.3mm]90:\tiny $\overline{q}_{i,0}$}] (u1) at (-0.5,1) {};
			\node[label={[label distance=-1.5mm]90:\tiny $\overline{q}_{i,j-1}$}] (u3) at (2,1) {};
			\node[label={[label distance=-1.5mm]90:\tiny $\overline{q}_{i,j,\ell'}$}] (u4) at (3.5,1) {};
			\node[label={[label distance=1mm]10:\tiny $\overline{q}_{i,j,\ell}$}] (u5) at (5,1) {};
			\node (u6) at (6.5,1) {};
			\node (u10) at (9,1) {};
			
			\draw[edge, loosely dotted] (u1) to (u3);
			\draw[edge, loosely dotted] (u3) to (u4);
			\draw[edge, loosely dotted] (u4) to (u5);
			\draw[edge, loosely dotted] (u5) to (u6);
			\draw[edge, loosely dotted] (u6) to (u10);
		
			\node (v1) at (-0.5,-0.5) {};
			\node (v3) at (2,-0.5) {};
			\node (v4) at (3.5,-0.5) {};
			\node (v5) at (5,-0.5) {};
			\node[label={[label distance=-0.5mm]270:\tiny $\underline{q}_{i,j}$}] (v6) at (6.5,-0.5) {};
			\node[label={[label distance=-0.9mm]270:\tiny $\underline{q}_{i,n_i}$}] (v10) at (9,-0.5) {};
			
			\draw[edge, loosely dotted] (v1) to  (v3);
			\draw[edge, loosely dotted] (v3) to (v4);
			\draw[edge, loosely dotted] (v4) to (v5);
			\draw[edge, loosely dotted] (v5) to (v6);
			\draw[edge, loosely dotted] (v6) to (v10);
			
			\draw[edge] (u1) to (v1);
			\draw[edge] (u3) to (v3);
			\draw[edge] (u4) to (v4);
			\draw[edge] (u6) to (v6);
			\draw[edge] (u10) to (v10);
			
			\draw[loosely dotted] (-0.5,1) .. controls (0.75,1) and (0.75,-0.5) .. (2,-0.5);
			\draw[loosely dotted] (6.5,1) .. controls (7.75,1) and (7.75,-0.5) .. (9,-0.5);
			\draw (2,1) .. controls (4.25,1) and (4.25,-0.5) .. (6.5,-0.5);

			
			\node[label={[label distance=0.7mm]135:\tiny $\overline{q}_{i',0}$}] (w1) at (0.25,-3) {};
			\node[label={[label distance=-1.5mm]135:\tiny $\overline{q}_{i',j'-1}$}] (w4) at (2.75,-3) {};
			\node[label={[label distance=.3mm]170:\tiny $\overline{q}_{i',j',\ell}$}] (w8) at (4.25,-3) {};
			\node[label={[label distance=-1.3mm]90:\tiny $\overline{q}_{i',j'}$}] (w9) at (5.75,-3) {};
			\node[label={[label distance=-2.2mm]90:\tiny $\overline{q}_{i',n_{i'}}$}] (w10) at (8.25,-3) {};
			
			\node (z1) at (0.25,-4.5) {};
			\node (z4) at (2.75,-4.5) {};
			\node (z8) at (4.25,-4.5) {};
			\node (z9) at (5.75,-4.5) {};
			\node (z10) at (8.25,-4.5) {};
		
			\draw[loosely dotted, edge] (w1) to (w4);
			\draw[loosely dotted, edge] (w4) to (w8);
			\draw[loosely dotted, edge] (w8) to (w9);
			\draw[loosely dotted, edge] (w9) to (w10);
		
			\draw[loosely dotted, edge] (z1) to (z4);
			\draw[loosely dotted, edge] (z4) to (z8);
			\draw[loosely dotted, edge] (z8) to (z9);
			\draw[loosely dotted, edge] (z9) to (z10);
			
			\draw[edge] (w1) to (z1);
			\draw[edge] (w4) to (z4);
			\draw (w9) to (z9);
			\draw[edge] (w10) to (z10);
			
			\draw[loosely dotted] (0.25,-3) .. controls (1.5,-3) and (1.5,-4.5) .. (2.75,-4.5);
			\draw[loosely dotted] (5.75,-3) .. controls (7,-3) and (7,-4.5) .. (8.25,-4.5);
			\draw (2.75,-3) .. controls (4.25,-3) and (4.25,-4.5) .. (5.75,-4.5);

		\end{pgfonlayer}
		
		
		\begin{pgfonlayer}{demands}
		
			\draw[edge, orange, line width=1pt,->,>= latex'] (sl) -- (u5);
			 \draw[edge, orange, line width=1pt,->,>= latex'](u5)-- (v5);
			 \draw[edge, orange, line width=1pt,->,>= latex'](v5)-- (w8); 
			 \draw[edge, orange, line width=1pt,->,>= latex'](w8)-- (z8); 
			 \draw[edge, orange, line width=1pt,->,>= latex'](z8)-- (tl);
				
			\draw[blue, line width=1pt] (u1.south) to (u3.south);
			\draw[blue, line width=1pt] (u3.south) .. controls (4.25,1) and (4.25,-0.5) .. (v6.south);
			\draw[edge, blue, line width=1pt] (v6.south) to (v10.south);
			
			\draw[edge, red, line width=1pt] (u1.north) -- (u4.north west) -- (v4.north west) -- (v10.north);
			
			\draw[edge, teal] (w1.north) to (w10.north);	

		\end{pgfonlayer}
		
		\draw [dashed] (-1.0,-1.75) -- (9.5,-1.75);
		
	\end{tikzpicture}
	
	\caption{A simplified schematic of construction in~\cite{amirimfcs}. The figure also contains examples of all possible types of paths that a solution can contain.}
	\label{fig:hardness}
\end{figure}
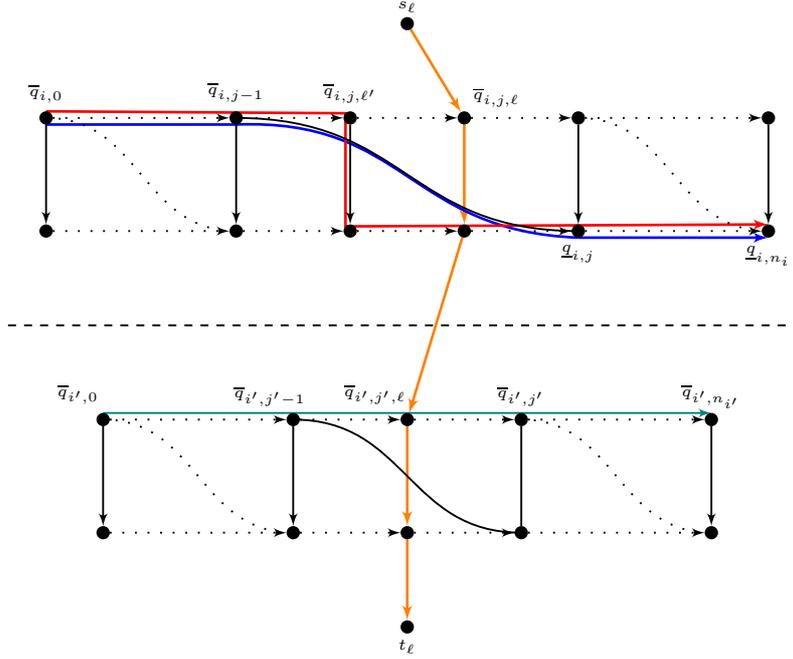

Given an input to the partitioned subgraph isomorphism problem as outlined above with $H$ being 3-regular and bipartite the authors construct their \disp{k}{c} instance as follows:

\subsubsection*{Hardness Construction of~\cite{amirimfcs}}

Let $V(H) = A \dot{\cup} B$ be a bipartition of $H$. Since $H$ is 3-regular it has to hold that $|A| = |B|$ and we can assume w.l.o.g.\ that $A = \{u_1, \ldots u_{h/2}\}$ and $B = \{u_{h/2 + 1}, \ldots, u_h\}$. For every vertex $u_i \in V(H)$ two directed paths $\overline{Q}_i$ and $\underline{Q}_i$ are created. We will only describe the construction of $\overline{Q}_i$, as the construction for $\underline{Q}_i$ is analogous. Fix an enumeration $V_i = \{v_{i,1}, \ldots, v_{i,n_i}\}$. Start with a directed path consisting of the vertices $\overline{q}_{i,0}, \ldots, \overline{q}_{i, n_i}$ in exactly this order. Subdivide this path by adding new vertices  $\overline{q}_{i,j,1}, \ldots, \overline{q}_{i,j,k}$ between consecutive vertices $\overline{q}_{i, j-1}$ and $\overline{q}_{i, j}$. These new vertices will serve as selection of $k$ edges of $H$. Both paths are then connected by adding all possible edges of the type $(\overline{q}_{i,j - 1}, \underline{q}_{i,j})$ and $(\overline{q}_{i,j,\ell}, \underline{q}_{i,j,\ell})$ to form a single \emph{block}. 

Next, link up these blocks by paths of length $5$: Fix an enumeration $e_1, \ldots, e_k$ of the edges in $H$. Given some edge $e_\ell = \{u_{i_1}, u_{i_2}\}$, introduce two new vertices $s_\ell, t_\ell$ and the following edges. Assume that $i_1 < i_2$; then for all $\{v_{i_1,j_1}, v_{i_2, j_2}\} \in E(G)$ add the edges $(s_\ell, \overline{q}_{i_1,j_1,\ell})$, $(\underline{q}_{i_1,j_1,\ell}, \overline{q}_{i_2,j_2,\ell})$, $(\underline{q}_{i_2,j_2,\ell}, t_\ell)$. 
All that is left is to add the demands:
\begin{itemize}
\item $c - 1$ copies of the demand $(\overline{q}_{i,0}, \overline{q}_{i, n_i})$,
\item $c - 1$ copies of the demand $(\underline{q}_{i,0}, \underline{q}_{i, n_i})$,
\item the demand $(\overline{q}_{i,0}, \underline{q}_{i, n_i})$.
\item the demand $(s_\ell, t_\ell)$ for all $1 \leq \ell \leq k$.
\end{itemize}

\noindent\textbf{Rationale behind the construction:}

The idea is that in every block $\overline{Q}_i \cup \underline{Q}_i$ the paths satisfying demands of the type $(s_\ell, t_\ell)$ can only be routed through a single window bordered by two consecutive vertices $\overline{q}_{i,j-1}$ and $\overline{q}_{i,j}$ for some $j$. This encodes a choice for the image of $u_i$ under $\varphi$ by setting $\varphi(u_i) = v_{i,j}$. Adjacency is preserved by construction and so a feasible solution for \disp{k}{c} gives a feasible solution for partitioned subgraph isomorphism and vice versa~\cite{amirimfcs}. Since the instance is of size $O(nk)$ and can be constructed in polynomial-time, all we have to do now, is to show that we can always modify solutions for this instance, so that all paths are actually shortest paths. 

\subsubsection*{From Disjoint Paths to Disjoint Shortest Paths}
With this intuition in mind, we can prove the~\Cref{thm:hardness}.

\begin{proof}[Proof of~\Cref{thm:hardness}]Consider a solution to the \disp{k}{c} instance described above. We show that we can modify such a solution to get one, in which all paths are shortest paths. In that case, we are done and Theorem \ref{thm:hardness} follows. 
	For each type of demand, we show that a path satisfying this demand either already is a shortest path or that it can be rerouted to yield a shortest path. Figure \ref{fig:hardness} schematically showcases all possible paths contained in a solution.
		
	\medskip
	1) Demands of type $(\overline{q}_{i,0}, \overline{q}_{i, n_i})$: There is exactly one path that satisfies this demand and so it necessarily has to be a shortest path (cf. the teal path in \ref{fig:hardness}). The same holds for lower blocking demands $(\underline{q}_{i,0}, \underline{q}_{i, n_i})$.
	
	\medskip
	2) Demands of type $(\overline{q}_{i,j - 1}, \underline{q}_{i, j})$. Here the graph contains paths connecting corresponding source and terminals that are not necessarily a shortest path (cf. the red path in~\ref{fig:hardness}). 
	Let's suppose the solution uses one such path $P$\footnote{In  fact, any solution to the given instance cannot use any of the red paths otherwise violates the congestion of some vertices. However, to avoid complications arising in explaining the details of the proof of previous work, we only show all paths in such a solution in polynomial time can be converted to shortest paths without violating the congestion criteria.}. 
	In that case we reroute $P$ and instead get a shortest path without increasing congestion of any vertex. If $P$ is not a shortest path, it has to contain an edge of type $(\overline{q}_{i,j,\ell}, \underline{q}_{i,j,\ell})$. We replace $P$ by a path $P'$, which goes from $\overline{q}_{i,0}$ to $\overline{q}_{i,j-1}$, takes the edge $(\overline{q}_{i,j-1}, \underline{q}_{i,j})$ and then continues onwards to $\underline{q}_{i,n_i}$ (cf. the blue path in~\ref{fig:hardness}). 
	This is a shortest path and replacing $P$ by $P'$ decreases the congestion of the vertices $\overline{q}_{i,j,1}, \ldots, \overline{q}_{i,j,\ell}, \underline{q}_{i,j,\ell}, \ldots, \underline{q}_{i,j,k}$, while leaving the others as before.
	
	\medskip
	3) Demands of type $(s_\ell, t_\ell)$: By construction, any $s_\ell$-$t_\ell$-path $P$ necessarily has to pass through two different blocks $\overline{Q}_i \cup \underline{Q}_i$ and $\overline{Q}_{i'} \cup \underline{Q}_{i'}$ with $1 \leq i \leq \frac{h}{2} < i' \leq h$ on its way to $t_\ell$; call these blocks $B_1,B_2$. Since an edge has to be chosen from each block it follows that $P$ has to have length at least 5. 	At the same time, however, we will show that $P$ cannot have a length larger than 5 either.
	
	We claim that $P$ cannot include any vertex of type $\overline{q}_{i,j}$ or $\underline{q}_{i,j}$, as there would be a vertex of congestion at least $c+1$ otherwise. Before proving the claim, let us see how to conclude that $P$ should have length exactly $5$, if the claim is correct. Since $P$ goes only through blocks $B_1,B_2$, it has to start with an edge from $s_\ell$ to $B_1$. Say that this edge ends up in $\overline{q}_{i,j,\ell}$. Similarly $P$ has to contain an edge from $B_2$ to $t_\ell$. Denote the tail of this edge by $\overline{q}_{i',j',\ell}$. Additionally $P$ has to contain another edge $e$ from $B_1$ to $B_2$. Since by above claim $P$ does not pass through any vertex of type $\overline{q}_{i,j}$ or $\underline{q}_{i,j}$ it means that $P$ has to connect $B_1$ and $B_2$ by the edge $(\underline{q}_{i,j,\ell},\overline{q}_{i',j',\ell})$. No other edges can be chosen. As such, $P$ has length exactly $5$ and since every such vertical path has to have length at least $5$, $P$ is a shortest path (cf.\ the orange path in~\ref{fig:hardness}).
		
	It remains to prove the claim. We only prove that $\overline{q}_{i,j}$ does not belong to $P$, as $\underline{q}_{i,j}$ can be dealt with analogously. We further assume that $\overline{q}_{i,j}$ is in block $B_1$ (block $B_2$ follows the same principle, we only have to have the name of the block for convenience in the proof).	Observe that due to the demands $(\overline{q}_{i,0}, \overline{q}_{i, n_i})$ and $(\underline{q}_{i,0}, \underline{q}_{i, n_i})$ from 1), every vertex contained in $\overline{Q}_i \cup \underline{Q}_i$ has to have congestion at least $c-1$.
	
	Now consider the path $P'$ satisfying the demand $(\overline{q}_{i,0}, \underline{q}_{i, n_i})$ from 2). W.l.o.g.\ (as we have seen in the previous case), we may assume $P'$ is  a shortest such path. $P'$ has to pass through $\overline{q}_{i,j}$ or $\underline{q}_{i,j+1}$ by construction; in the former case, we are done, since the congestion of $\overline{q}_{i,j}$ will be $c$, thus it cannot be part of $P$. It remains to show that if $P'$ intersects $\underline{q}_{i,j+1}$, we cannot have $\overline{q}_{i,j}$ in $P$.
	
	For the sake of contradiction, suppose $P$ goes through $\overline{q}_{i,j}$. Denote the first vertex of $\underline{Q}_i \cap P'$ w.r.t.\ the topological order of the graph by $\underline{q}_{i,j'}$. Clearly every vertex on $\underline{Q}_i \cap P'$ has to have congestion exactly $c$. If $j'\le j$, then since $P$ should leave the block $B_1$ from one of the vertices in $\underline{Q}_i$ and since every such vertex has to appear after $\overline{q}_{i,j}$ in the topological order of the graph, they already have congestion $c$ and they cannot be part of $P$, a contradiction. On the other hand if $j'>j$, then $\overline{Q}_i \cap P'$ already includes the vertex $\overline{q}_{i,j}$ thus its congestion will be $(c-1)+1+1 = c+1$; a contradiction.
\end{proof}

\medskip
\noindent\textbf{Hardness of $k$-Edge Disjoint Shortest Paths ($k$-EDSP) Problem on Planar DAGs}
A natural limitation is to consider the problem on more restricted graph classes, in particular on planar graphs. One can observe that our algorithm with slight modification works for the $k$-EDSP on DAGs. It is natural to ask whether the problem on planar DAGs is easier? We answer this question negatively in the case of congestion one. In general undirected graphs, the vertex variant of the problem is known to be W[1]-hard by the construction provided in~\cite{undirectedshortest2}. They used the colored clique problem to prove this claim. By using the colored clique problem we show there is no $n^{o(k)}$ algorithm for $k$-EDSP on planar acyclic graphs unless ETH fails.

It is known that there does not exist any algorithm that decides the existence of a clique of size $k$ in a graph of size $n$ in time $f(k)n^{o(k)}$, unless all problems in SNP can be solved in subexponential time \cite{chen} (In particular ETH would have to fail). Consider now the \textsc{Multi-Colored Clique} problem. In it one has to decide, for a given graph $G$ with coloring $c: V(G) \to \{1, \ldots, k\}$, whether there exists a colorful clique in $G$ on $k$ vertices.

The hardness result from \cite{chen} directly applies to the \textsc{Multi-Colored Clique} problem via a simple and well-known reduction. Given an instance $(G,k)$ with $n$ vertices of \textsc{Clique} construct an instance $(G', c)$ of \textsc{Multi-Colored Clique} as follows: for every vertex $v \in V(G)$ add $k$ vertices $v_1, \ldots, v_k$ to $G'$ and set $c(v_i) = i$. Given two different vertices $u, v \in V(G)$, add all possible edges $\{u_i, v_j\}$ to $G'$, if $\{u,v\} \in E(G)$. Clearly, $(G,k)$ has a solution, if and only if $(G', c)$ has one. We show that this hardness result also carries over to the $k$-Edge Disjoint Shortest Paths problem on planar DAGs. In particular we prove the following theorem.

\begin{apthm}{thm:sdpplanar}
Unless ETH fails, the $k$-EDSP problem does not admit any $f(k)n^{o(k)}$ algorithm even on planar DAGs.
\end{apthm}

To do so, we extend the hardness proof of $k$-Vertex Disjoint Shortest Paths on general undirected graphs that appeared in~\cite{undirectedshortest2} to $k$-Edge Disjoint Shortest Paths on Planar DAGs.
\begin{proof}
	Let $I=(G, c,k)$ be an instance of the \textsc{Multi-Colored Clique} problem. Let $V(G) = \{v_1, \ldots, v_n\}$ and w.l.o.g.\ we assume that $c(v_i) \leq c(v_j)$, if $i < j$. We construct an instance $I'=(G',\{(s_1,t_1,\ldots,s_{2k},t_{2k})\})$  of $2k$-EDSP problem on planar DAGs as follows. 

Start with an $n \times n$ planar directed grid, that is we direct edges of an undirected grid from the left to the right and from the top to the bottom. Denote the vertex that is placed on the intersection of the $i$'th row and the $j$'th column by $w^{\text{out}}_{i,j}$. Surrounding this grid we place 2 other vertices $\tilde{s}^\star_i$ and $\tilde{t}^\star_i$ for all indices $1 \leq i \leq n$ 
where $\star \in \{\text{h}, \text{v}\}$ is the direction. Additionally, introduce edges $(\tilde{s}^h_i, w^{\text out}_{i,1}), (w^{\text out}_{i,n}, \tilde{t}^h_i), (\tilde{s}^v_i, w^{\text out}_{1,i}), (w^{\text out}_{n,i}, \tilde{t}^v_n)$. Clearly, the constructed graph so far is a planar DAG. 

\begin{figure}
	\input{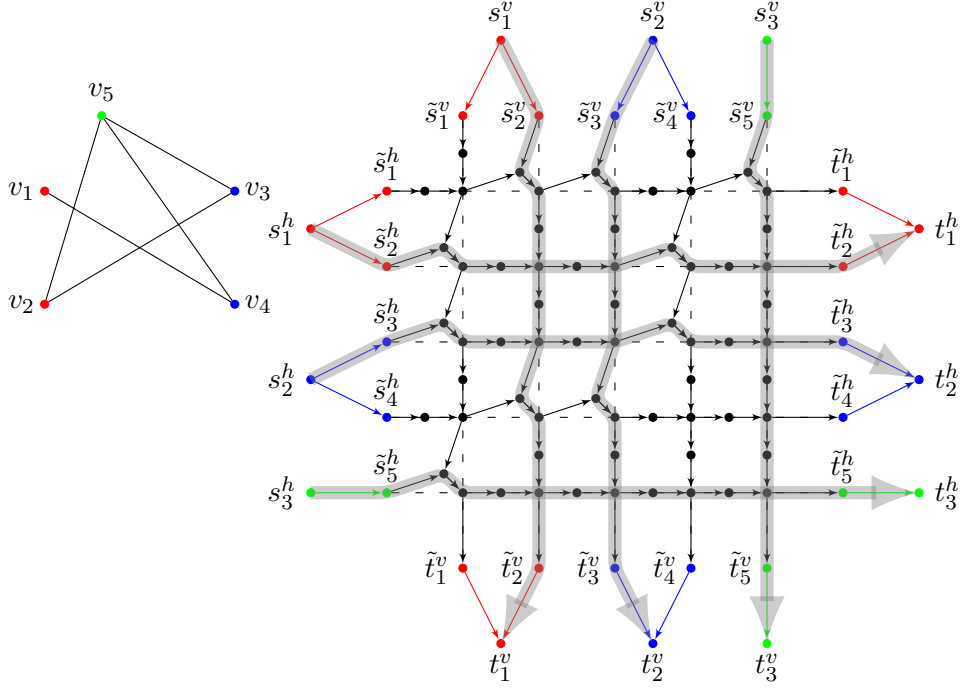}
	\caption{An instance of colored clique problem is given at left, a corresponding edge-disjoint shortest paths instance on a planar DAG is drawn in the right. The highlighted gray edges showing the solution to the disjoint shortest paths and consequently to the colored-clique problem: the graph induced on the corresponding vertices of the selected paths, in this case, vertices with indices $\{2,3,5\}$, is a colorful clique on $k$ vertices. Dashed lines are showing the initial grid-like structure.}\label{fig:edgedisjoint}
\end{figure}

Now, split up every edge and add a new vertex in between: i.e.\ a horizontal edge $(w^{\text out}_{i,j - 1}, w^{\text out}_{i,j})$ is replaced with a directed path of length two with edges $(w^{\text out}_{i,j-1}, w^{\text in}_{i,j - \varepsilon})$ and $(w^{\text in}_{i,j - \varepsilon}, w^{\text out}_{i,j})$, where $w^{\text in}_{i-1,j-\varepsilon}$ is the new vertex. Similarly, a vertical edge $(w^{\text out}_{i-1,j}, w^{\text out}_{i,j})$ is replaced by $(w^{\text out}_{i-1,j}, w^{\text in}_{i - \varepsilon,j})$ and $(w^{\text in}_{i - \varepsilon,j} , w^{\text out}_{i,j})$, where $w^{\text in}_{i - \varepsilon,j}$ is the new vertex. 
	For all edges $(\tilde{s}^h_i, w^{\text out}_{i,1})$ and $(\tilde{s}^v_i, w^{\text out}_{1,i})$ we do the same and introduce the vertices $w^{\text in}_{i,1-\varepsilon}$ and $w^{\text in}_{1-\varepsilon, i}$, respectively.

	Next, we merge some of the vertices, while keeping the graph planar and acyclic. That is, we merge all vertices $w^{\text in}_{i-\varepsilon,j}$ and $w^{\text in}_{i, j-\varepsilon}$ with $i \neq j$, unless $v_i$ and $v_j$ differ in color and are adjacent in $G$. Call the resulting vertex $w^{\text in}_{i,j}$. Notice, that the length of the rows and columns stays the same, thus, they all have the same length (cf. Figure~\ref{fig:edgedisjoint}). 
	
	To finish the graph construction we add $2$ additional vertices $s^\star_\ell, t^\star_\ell$ for all indices $\ell\in [k]$ and directions $\star\in\{h,v\}$ and connect them to the rest of the graph as follows: If $c(v_i) = \ell$, then add the edges $(s^h_\ell, \tilde{s}^h_i), (s^v_\ell, \tilde{s}^v_i), (\tilde{t}^h_i, t^h_\ell), (\tilde{t}^v_i, t^v_\ell)$. Finally, define demands $(s^h_\ell, t^h_\ell)$ and $(s^v_\ell, t^v_\ell)$ for all $\ell \in [k]$ to finish construction of instance $I'$. ~\Cref{fig:edgedisjoint} illustrates the construction. 
	
	\textsc{Multi-Colored Clique $\rightarrow$ $2k$-EDSP:} Let $v_{i_1}, \ldots, v_{i_k}$ be the solution to $I$, with $c(v_{i_j}) = j$. Then $I'$ has a solution as follows: It consists of the paths that go from $s^h_j$ (resp.\ $s^v_j$) to $t^h_j$ (resp. $t^v_j$) via the $i_j$'th row (resp.\ column).
	
	These paths can only share an edge, if they are \emph{orthogonal} to one another, that is, if one is vertical and the other horizontal.
	Take two such paths, say the horizontal path corresponding to $v_{i_j}$ and the vertical path corresponding to $v_{i_\ell}$. The only edge they could in principle share is $(w^{\text in}_{i_j, i_\ell}, w^{\text out}_{i_j, i_\ell})$. However, $v_i$ and $v_j$ have to be adjacent in $G$ and cannot have the same color. But then $w^{\text in}_{i_j,i_\ell}$ has never been formed and so the paths have to be disjoint. Since all of them are shortest paths (they are rows/columns in the grid), they form a valid solution for $I'$.

	For the converse implication, i.e.\ \textsc{$2k$-EDSP $\rightarrow$ Multi-Colored Clique}, suppose that there is a solution $P^h_1, \ldots, P^h_k, P^v_1, \ldots, P^v_k$ for $I'$, where $P^\star_\ell$ (for $\ell \in [k], \star\in\{h,v\}$), satisfies the demand $(s^\star_\ell, t^\star_\ell)$. Since these are shortest paths they cannot leave their respective row or column.
	Let's say $P^h_\ell$ uses the $i_\ell$'th row, then by construction, $P^v_\ell$ has to use the $i_\ell$'th column, since every other vertical shortest path (that is, every other column) for the color $\ell$ shares an edge with $P^h_\ell$. 
	 We now want to show that the corresponding vertices $v_{i_1},\ldots,v_{i_k}$ form a solution for the \textsc{Multi-Colored Clique} instance. 
	
	It is clear that $c(v_{i_\ell}) = \ell$. Furthermore, all of the vertices have to be connected, as otherwise their corresponding paths would share an edge. Say, there is no edge between $v_{i_\ell}$ and $v_{i_j}$, then $P^h_\ell$ and $P^v_j$ would share the edge $(w^{\text in}_{i_\ell, i_j}, w^{\text out}_{i_\ell, i_j})$: both $w^{\text in}_{i_\ell - \varepsilon, i_j}$ and $w^{\text in}_{i_\ell, i_j-  \varepsilon}$ would have been merged. As such, $v_{i_1}, \ldots, v_{i_k}$ forms a solution to the \textsc{Multi-Colored Clique} instance.
	Since the $2k$-EDSP instance was constructed in polynomial time and has $O(n^2 + k)$ vertices, the theorem follows. 
\end{proof}

	\clearpage
	\section{Conclusion and Open Problems}
We introduced the theoretical study of shortest disjoint paths with congestion. The problem is practically relevant since in real networks congestion is inevitable, moreover, also we are interested in short routes. From a theoretical perspective, the problem is derived from two already well-known problem sets of routing with congestion and routing shortest paths.	We have provided algorithms and hardness results mostly centered around acyclic graphs. Since the concept is new, there are many open problems ahead.
	
The major open problem is to find out if the \dsp{k}{c} problem in general graphs admits an algorithm with running time $f(k)n^{g(d)}$. Approaching this problem might be a bit ambitious, since the complexity of the simpler problem of \sd{k} is wide open in digraphs, even for a restricted case of $k=3$. 
	
	We have presented an algorithm for undirected graphs with running time of form $n^{f(k)}$. An interesting intermediate problem is  to solve \dsp{k}{c} in undirected graphs in time $f(k)n^{g(d)}$.  
	
Even though the running time of our algorithm is close to the lower bound, it is not clear if it is possible to change the reduction step so that it requires routing fewer paths. More generally is there an algorithm for DAGs with run time $O(f(k) n^{d+c_1})$, for some constant $c_1$? 
		
We showed the edge-disjoint variant of the problem on planar DAGs is W[1]-hard, which means the algorithm of Kobayashi and Sako~\cite{shortest2} cannot turn to FPT under ETH assumption, but it is not clear if the vertex disjoint variant of the problem is hard on planar graphs.
	
	\medskip
	
\textbf{Acknowledgements. }We would like to thank Sebastian Siebertz for his valuable feedback on the draft version of this paper. In addition, we would like to thank Rajesh Chitnis, Mohammad Roghani, and Reza Soltani for fruitful discussions on the disjoint shortest paths problem.
\clearpage
		\bibliographystyle{plain}	\bibliography{ref}  

	\end{document}